\documentclass[11pt,a4paper]{article}

\usepackage{amsmath}
\usepackage{amsfonts}
\usepackage{amsthm}

\newtheorem{Theorem}{Theorem}

\newtheorem{Definition}{Definition}
\newtheorem{Proposition}{Proposition}
\newtheorem{Lemma}{Lemma}
\newtheorem{Remark}{Remark}

\newcommand{\LL}{{\mathrm{L}}}
\newcommand{\CC}{{\mathrm{C}}}
\newcommand{\Id}{{\mathbf{1}}}
\newcommand{\dom}{{\mathrm{dom}~}}
\newcommand{\img}{{\mathrm{rng}~}}

\newcommand{\R}{\ensuremath{{\mathbb R}}}
\newcommand{\C}{\ensuremath{{\mathbb C}}}

\newcommand{\Imm}{{\mathrm{Im}~}}
\newcommand{\Ree}{{\mathrm{Re}~}}
\newcommand{\hil}{\mathcal H}
\newcommand{\ra}{\rangle}
\newcommand{\la}{\langle}
\newcommand{\wgconv}{\stackrel{\mathrm{W\Gamma}}{\longrightarrow}}
\newcommand{\sgconv}{\stackrel{\mathrm{S\Gamma}}{\longrightarrow}}
\newcommand{\wseta}{\rightharpoonup}
\newcommand{\var}{\varepsilon}

\newcommand{\too}{\rightarrow}

\begin{document}

\title{Complex $\Gamma$-convergence and magnetic Dirichlet Laplacian in bounded thin tubes}
\author{R. Bedoya, C. R. de Oliveira { \small and} A. A. Verri\\ 
\vspace{-0.6cm}
\small
\em Departamento de Matem\'{a}tica -- UFSCar, \small \it S\~{a}o Carlos, SP,
13560-970
Brazil\\ \\}
\date{\today}

\maketitle 

\begin{abstract}  
The resolvent convergence of self-adjoint operators via the technique of  $\Gamma$-convergence of quadratic forms is adapted to incorporate complex Hilbert spaces. As an  application, we find effective operators to the Dirichlet Laplacian with magnetic potentials in very thin bounded tubular regions in space  built along smooth closed curves;  relatively weak regularity is asked for the potentials,  and the convergence is in the norm resolvent sense as the cross sections of the tubes go uniformly to zero.
\end{abstract}

\section{Introduction}

Consider a family of (lower bounded) self-adjoint operators $T_\varepsilon$ ($\var > 0$) with domain $\dom T_\varepsilon$ in a complex separable Hilbert space~$\hil$, and the corresponding closed sesqui\-lin\-e\-ar forms $b_\varepsilon$. We want to study the limit~$T$ (resp.\  $b$) of $T_\varepsilon$ (resp.~$b_\varepsilon$) as~$\var \too 0$.  We can relate this study to the concept of $\Gamma$-con\-ver\-gence; but when addressing sesqui\-lin\-e\-ar forms these variational problems are usually formulated in real Hilbert spaces and the theory has been developed under this condition; see, for instance, some important  monographs in the field as~\cite{DalMaso,Braides}. However, in quantum mechanics the Hilbert spaces are usually complex, and here we have the first aim of this work, that is, to shortly explain the ideas involved in the $\Gamma$-con\-ver\-gence of quadratic forms and then describe the modifications needed to generalize the pertinent results to complex Hilbert spaces. It must be underlined that this adaptation to complex Hilbert spaces do not make things easier, but it is handy to accommodate quantum mechanics, in particular if magnetic potentials are present. As an application we study the effective operator obtained from the Dirichlet Laplacian with magnetic field restricted to closed bounded tubes in~$\R^3$ that shrinks to a smooth curve.

There are several papers on the  Dirichlet Laplacian without  magnetic fields restricted to tubes in~$\mathbb R^3$; see, for instance, \cite{BMT, FS1, FS2, CRO, AAV, OCA}. In particular, the variational technique of $\Gamma$-convergence in real Hilbert spaces  was invoked in~\cite{BMT, CRO}.  One of the first works in which a magnetic field was added to this kind of problem is the paper~\cite{Gru}, where  the author has obtained an asymptotic expansion of the eigenvalues; but special particularities about the field were imposed. In that work the problem was restricted to a sequence of bounded tubes $\Lambda_\var$  of the space that shrinks to a closed curve  of $\mathbb R^2$ as $\var \too 0$.  Recently, considering  now that the tubes $\Lambda_\var$ are unbounded,  it was proven in~\cite{DKRAY} the norm resolvent convergence  under the condition that the vector field ${\bf A}$ depends on a parameter, more precisely, it is of form $b {\bf A}$, where $b$ is a positive parameter that depends on~$\var$.   Other variation of this problem was studied in~\cite{DKNRMT}, where the authors have considered the Dirichlet Laplacian between two parallel hypersurfaces in  Euclidean space in the presence of a magnetic field. When the distance between them tends to zero, it was shown a norm resolvent convergence of the associated operators. Since we will make use of the $\Gamma$-convergence, we will be able to require  weak regularity of the magnetic potential.

Let~$S$ be a circle of length $l>0$ and $r: S \too \mathbb R^3$ a closed and simple curve of class~$\CC^3$ in~$\mathbb R^3$ parameterized by its arc length~$s$. Denote by $k(s)$ and $\tau(s)$  its curvature and torsion at the point~$r(s)$, respectively.  Let~$Q$ be a smooth open, bounded, simply connected,  and nonempty subset of $\mathbb R^2$. We build a tube $\Omega$ in~$\mathbb R^3$ by moving the region~$Q$ along~$r(s)$.  At each point the region may present an additional rotation angle  which is denoted by~$\alpha(s)$ and we suppose that its of class $\CC^2$, and  the Dirichlet condition at the boundary~$\partial \Omega$. We take a vector magnetic potential field ${\bf A} = (A_1, A_2, A_3)$, where $A_j :\Omega \too \mathbb R$, $j=1,2,3$, are real functions so that, for differentiable~${\bf A}$, ${\bf B} = \nabla\times\textbf{A}$ is the corresponding magnetic field. Consider  the family of operators 
$$(H_\var \psi)(x)  := \left[(-i \partial - {\bf A})^2 \psi\right](x) \qquad (0 < \var < 1),$$
$\dom H_\var = \hil^2(\Omega_\var) \cap \hil_0^1(\Omega_\var)$ ($\Omega_\var$ is the region obtained by moving the $\var Q$ along~$r(s)$ and, for each $0< \var < 1$, we consider ${\bf A}$ restricted to $\Omega_\var$); see around~\eqref{condA} for the regularity conditions imposed on~${\bf A}$. 

We study the sequence $H_\var$ in the limit $\var \too 0$. For this it is necessary to make some renormalization, for example, we  need to control  the transverse oscillations as $\var \too 0$. An interesting point  is that even in the presence of a vector potential we are going to control these oscillations by  subtracting $\lambda_0/\var^2$ from $H_\var$, where~$\lambda_0$ is the first (i.e., the lowest) eigenvalue of the Dirichlet Laplacian (no magnetic potential!) restricted to~$Q$. Namely,
\begin{equation}\label{crosssectioneigenvalue}
-\Delta u_0 = \lambda_0 u_0,
\quad u_0 \in \hil_0^1(Q), \quad u_0 \ge 0,
\quad \int_Q |u_0|^2 {\mathrm d}y = 1. 
\end{equation}
$u_0$ denotes the normalized eigenfunction associated with~$\lambda_0$. Recall that~$\lambda_0>0$ and  it is a simple eigenvalue.

Now, consider the one-dimensional operator 
$$(G_0w)(s):= \big(-i \partial_s - \langle{\bf A}(r(s)), T(s) \rangle \big)^2 w(s) + \left[C(Q) (\tau+\alpha')^2(s) - \frac{k^2(s)}{4} \right]w(s),$$
$\dom G_0 = \hil^2(S)$, where~$C(Q)$ is a number that  depends only on  the region~$Q$ (see~\eqref{constC}) and $T(s)$ is the tangent vector to the curve~$r$ at the position~$r(s)$; $\langle\cdot,\cdot\rangle$ denotes the usual scalar product in~$\mathbb R^3$. Our main application here says that
\begin{equation}\label{resultintroductionorm}
H_\var - \frac{\lambda_0}{\var^2}\Id  \underset{\var \too 0}{\longrightarrow} G_0
\end{equation}
in a norm-resolvent sense; see Theorem~\ref{teoremaprincipal} for a precise formulation. To prove this convergence  we  use the variational technique of  $\Gamma$-convergence of quadratic forms in complex Hilbert spaces. Thus, in Section~\ref{secIntro} we make the necessary  generalizations to complex spaces in order to combine the corresponding strong and weak $\Gamma$-convergences with the operator convergence~\eqref{resultintroductionorm}. In Section~\ref{themodel} we show some steps of the construction of the region where the problem is considered, the quadratic forms, and appropriate change of variables and renormalization. We also comment about a suitable  gauge transform related to the magnetic potential~${\bf A}$. In Section~\ref{maintheorem} we present the main results related to our application.

\section{$\Gamma$-convergence in complex Hilbert spaces}\label{secIntro} 

As already mentioned in the Introduction, we consider a family of (uniformly) lower bounded self-adjoint operators $T_\varepsilon$ ($\var > 0$) with domain $\dom T_\varepsilon$ in a complex separable Hilbert space~$(\hil,\langle\cdot,\cdot\rangle)$. We denote by~$b_\varepsilon$ the corresponding closed sesqui\-lin\-e\-ar forms, and want to study the limit~$T$ (resp.\  $b$) of $T_\varepsilon$ (resp.~$b_\varepsilon$) as $\var \too 0$. The domain of~$T$ will not be supposed to be dense
in~$\hil$ and its closure will be denoted by $\hil_0=\overline{\dom T}$ (with $\img T\subset
\hil_0$); usually this is indicated by simply saying that ``$T$ is self-adjoint in~$\hil_0$.''
As usual, the real-valued function $\zeta\mapsto b(\zeta,\zeta)$ will be simply
denoted by
$b(\zeta)$ and called the associated quadratic form with~$b(\zeta,\eta)$. It will  be assumed that $b$ is positive (or lower bounded in general) and
$b(\zeta)=\infty$ if $\zeta$ does not belong to its domain
$\dom  b$; this is important in order to guarantee that in some cases $b$ is lower semicontinuous,
which is equivalent to $b$ be the sesqui\-lin\-e\-ar form generated by a positive self-adjoint
operator~$T$, that is, 
\[ b(\zeta,\eta)=\la T^{1/2}\zeta,T^{1/2}\eta\ra, \quad \zeta,\eta\in\dom b=\dom T^{1/2};
\]see Theorem 9.3.11 in~\cite{ISTQD}. By allowing $b(\zeta)=\infty$ one has a handy way to work in the larger space~$\hil$ instead of only in~$\hil_0=\overline{\dom T}$. 

\begin{Definition}\label{GammaAltern} The sequence $f_\varepsilon:\hil\to \overline{\R}$ strongly
$\Gamma$-converges to $f$ (that is, $f_\varepsilon\sgconv f$) iff the following two conditions are
satisfied:
\begin{itemize}
\item[i)] For every $\zeta\in\hil$ and every $\zeta_\varepsilon\to\zeta$ in~$\hil$ one has
\[ f(\zeta)\le \liminf_{\varepsilon\to0} f_\varepsilon(\zeta_\varepsilon).
\]
\item[ii)] For every $\zeta\in\hil$ there exists a sequence $\zeta_\varepsilon\to\zeta$ in
$\hil$ such that
\[ f(\zeta)=\lim_{\varepsilon\to0} f_\varepsilon(\zeta_\varepsilon).
\]
\end{itemize}
\end{Definition}

\begin{Remark}\label{remark3} If instead of strong con\-ver\-gence $\zeta_\varepsilon\to\zeta$ one
considers weak con\-ver\-gence $\zeta_\varepsilon\wseta\zeta$ in Definition~\ref{GammaAltern}, then one has a characterization of $f_\varepsilon\wgconv f$, that is, $f_\var$ weakly $\Gamma$-converges to~$f$. 
\end{Remark}

Now we state, in an appropriate form, the main result relating strong resolvent con\-ver\-gence of self-adjoint operators and $\Gamma$-con\-ver\-gence of the associated  sesqui\-lin\-e\-ar forms.  
\begin{Theorem}\label{mainTheorGamma} Let $b_\varepsilon,b$ be positive (or uniformly lower bounded)
closed sesqui\-lin\-e\-ar forms in the complex Hilbert space~$\hil$, and
$T_\varepsilon,T$ the corresponding associated positive self-adjoint operators. Then the following
statements are equivalent:
\begin{itemize}
\item[i)] $b_\varepsilon\sgconv b$ and, for each $\zeta\in\hil$, $b(\zeta)\le
\liminf_{\varepsilon\to0} b_\varepsilon(\zeta_\varepsilon)$, $\forall \zeta_\varepsilon\wseta
\zeta$ in~$\hil$.
\item[ii)]$b_\varepsilon\sgconv b$ and $b_\varepsilon\wgconv b$.
\item[iii)] $b_\varepsilon+\lambda \sgconv b+\lambda$ and $b_\varepsilon+\lambda \wgconv b+\lambda$,
for some
$\lambda>0$ (and so for all $\lambda\ge0$).
\item[iv)] $T_\varepsilon$ converges to~$T$ in the strong resolvent sense in
$\hil_0=\overline{\dom T}\subset\hil$, that is, 
\[
\lim_{\varepsilon\to0} R_{-\lambda}(T_\varepsilon)\zeta = R_{-\lambda}(T)P_0\zeta,\quad
\forall\zeta\in\hil,\forall \lambda>0,
\]where $P_0$ is the orthogonal projection onto~$\hil_0$.
\end{itemize}
\end{Theorem}

Now we provide the necessary  modifications so that the proofs of Theorem~13.6 and Corollary~13.7 in~\cite{DalMaso} can be replicated in order to include the case of complex Hilbert spaces, and so to conclude Theorem~\ref{mainTheorGamma} above. There are two main points to regard. The first one is the replacement, in many instances, of terms of the form $2\la\eta,\cdot\ra$ in real-space functionals by $\la\eta,\cdot\ra+\la\cdot,\eta\ra$; although this substitution is quite natural, there are few nuances in the proofs (see Proposition~\ref{propMudancaComplex} ahead). Further, the proofs also help to elucidate the connection between forms, operator actions  and domains on the one hand, and minimalization of suitable functionals on the other hand; this sheds some light on  the role played by $\Gamma$-con\-ver\-gence in the con\-ver\-gence of self-adjoint operators. 

\begin{Proposition}\label{propMudancaComplex} Let $b\ge0$ be a closed sesqui\-lin\-e\-ar form in the complex Hilbert space~$\hil$, $T\ge0$ the self-adjoint operator associated with $b$ and~$P_0$ be the orthogonal projection onto $\hil_0=\overline{\dom T}\subset\hil$. Then $\zeta\in\dom T$ and $T\zeta=P_0\eta$ iff $\zeta$ is a minimum point (also called minimizer) of the functional
\[ g:\hil\to\overline\R,\qquad g(\zeta)= b(\zeta) - \la\eta,\zeta\ra-\la\zeta,\eta\ra.
\]
\end{Proposition}

\begin{proof} Assume that $\zeta\in\dom T$ and $T\zeta=P_0\eta$. Note that $g(0)=0$, so that the
minimum of
$g$ is $<\infty$. If $\varrho\in\hil\setminus \dom b$, then $g(\varrho)=\infty$ so that we can assume
that
$\varrho\in\dom b$; thus $\phi=\varrho-\zeta\in\dom b\subset\hil_0$ and
\begin{eqnarray*}
 b(\varrho)&=& b(\zeta+(\varrho-\zeta))=b(\zeta)+b(\phi)+b(\zeta,\phi) + b(\phi,\zeta)
\\ &\ge& b(\zeta) + \la T\zeta,\phi\ra +\la\phi,T\zeta\ra
\\ &=& b(\zeta) +  \la P_0\eta,\phi\ra +\la\phi,P_0\eta\ra
\\ &=& b(\zeta) + \la \eta,\phi\ra +\la\phi,\eta\ra
\\ &=& b(\zeta) + \la\eta,\varrho-\zeta\ra + \la \varrho-\zeta,\eta\ra.
\end{eqnarray*} Hence
$b(\zeta)- \la\eta,\zeta\ra - \la\zeta,\eta\ra \le b(\varrho) - \la \eta,\varrho\ra  -
\la\varrho,\eta\ra,
$ which is equivalent to $g(\zeta)\le g(\varrho)$. Since $\varrho$ was arbitrary, $\zeta$ is a minimum
point of~$g$.  Suppose now that $g(\zeta)\le g(\varrho)$, for all  $\varrho\in\hil$, that is, $\zeta$
is a minimum point of~$g$. Since $g(\zeta)\le g(0)=0$, it follows that $0\le b(\zeta)\le
\la\eta,\zeta\ra +\la\zeta,\eta\ra<\infty$ and so
$\zeta\in\dom b$. The hypothesis
$g(\zeta)\le g(\varrho)$ amounts to
\[ b(\varrho) \ge b(\zeta) + \la\eta,\varrho-\zeta\ra + \la\varrho-\zeta,\eta\ra,\quad \forall
\varrho.
\] Now for $\varphi\in\dom b$ and $z\in \C$ fixed, this inequality implies
\begin{eqnarray*} b(\zeta) + |z|^2 b(\varphi) + zb(\zeta,\varphi) +\overline z b(\varphi,\zeta)&=&
b(\zeta+z\varphi) 
\\ &\ge& b(\zeta) + z\la\eta,\varphi\ra + \overline z\la\varphi,\eta\ra.
\end{eqnarray*}
 Choosing $z=t>0$ yields
 \[
 tb(\varphi) + b(\zeta,\varphi)+b(\varphi,\zeta) \ge \la \eta,\varphi\ra + \la \varphi,\eta\ra,
 \]and taking $t\to 0$ one finds 
 \[
 b(\zeta,\varphi)+b(\varphi,\zeta) \ge \la \eta,\varphi\ra + \la \varphi,\eta\ra.
 \]
 
 By considering $z=t<0$ and then $t\to0$ one gets the opposite inequality, and so the first relation
 \[
 b(\zeta,\varphi)+b(\varphi,\zeta)= \la \eta,\varphi\ra + \la \varphi,\eta\ra.
 \] By taking successively $z=it$ with $t>0$ and $t<0$, then $t\to0$ in both cases, one gets the
second relation
\[
 b(\zeta,\varphi)-b(\varphi,\zeta)= \la \eta,\varphi\ra - \la \varphi,\eta\ra.
 \]  Add these two relations to obtain
\[ b(\zeta,\varphi) = \la \eta,\varphi\ra = \la P_0\eta,\varphi\ra, \quad \forall
\varphi\in\dom b,
\]and so conclude that $\zeta\in\dom T$ and $T\zeta=P_0\eta$ (see page~101 in~\cite{ISTQD}). This finishes the proof of the proposition.
\end{proof}

The second main technical point we need to prove Theorem~\ref{mainTheorGamma} in the case of complex
Hilbert spaces is the following complex version of Proposition~11.9 in~\cite{DalMaso}. 
\begin{Proposition}\label{ComplexQF} 
Let~$\hil$ be a complex Hilbert space and
$F:\hil\to[0,\infty]$. If this functional $F$  satisfies
\begin{itemize}
\item[a)] $F(0)=0$,
\item[b)] $F(t\zeta)\le t^2F(\zeta)$, for all $\zeta\in\hil$ and all $t\ge0$,
\item[c)] $F(\zeta+\eta) + F(\zeta-\eta) \le 2F(\zeta)+2F(\eta),$ for all $\zeta,\eta\in\hil$,
\item[d)] $F(i\zeta)=F(\zeta)$, for all $ \zeta\in\hil$,
\end{itemize} 
then $F$ is a quadratic form on~$\hil$. Conversely, if $F:\hil\to[0,\infty]$ is a
quadratic form, then it satisfies a),b),c),d) and, in addition,
\begin{itemize}
\item[{\it e)}] $F(z\zeta) = |z|^2F(\zeta)$, for all $\zeta\in\hil$ and all $z\in\C$,
\item[{\it f)}] $F(\zeta+\eta) + F(\zeta-\eta)= 2F(\zeta)+2F(\eta),$ for all $ \zeta,\eta\in\hil$.
\end{itemize}
\end{Proposition}
\begin{proof} 
If we first restrict $F$ to real scalars and keep the notation $F$ (so {\it d)} becomes
meaningless), then by Proposition~11.9 of~\cite{DalMaso} (since {\it a)}, {\it b)}, {\it c)} hold true) this restriction is
the quadratic form associated with a real sesqui\-lin\-e\-ar form $B:Y\times Y\to \R$,
$Y=\{\zeta\in\hil: F(\zeta)<\infty\}$, given by
\[ B(\zeta,\eta)=\frac14\left( F(\zeta+\eta)-F(\zeta-\eta) \right).
\] Note that $B(\zeta):=B(\zeta,\zeta)=F(\zeta)$ and, in particular,
$B(\zeta,\eta)=B(\eta,\zeta)$, $B(t\zeta,s\eta)=tsB(\zeta,\eta),$ for all $ \zeta,\eta\in Y$, for all
$ t,s\in\R$. Further, except \rm{d)}, $B(\zeta)$ satisfies all items
{\it a)},$\cdots$,{\it f)} but with the restriction $z\in\R$ in {\it e)}.   

Our task now is to introduce an
appropriate complex version of~$B$, also defined on~$Y$. To this end first extend~$B$ by considering the original~$F$ (i.e., without the restriction to real scalars) in the above expression, then define
\[ b(\zeta,\eta):= B(\zeta,\eta)-i B(\zeta,i\eta),
\] and we will check that it works, that is, that $b$ is a (complex) sesqui\-lin\-e\-ar form and $b(\zeta)=F(\zeta)$, for all~$\zeta$.    The motivation for this expression for $b$ comes from  the following remark: if $u:Y\to\C$ is a linear functional, then $u(i\zeta)=\Ree u(i\zeta)+i\Imm u(i\zeta)=iu(\zeta)=i\Ree u(\zeta)-\Imm u(\zeta)$,  that is,  the
relation $\Imm u(\zeta)=-\Ree u(i\zeta)$ is valid (recall that $b$ must be linear in the second variable).  By \rm{d)} and the definitions of~$B$ and $b$ it follows that
\begin{itemize}
\item $B(i\zeta,i\eta)=B(\zeta,\eta)$ and $b(i\zeta,i\eta)=b(\zeta,\eta)$, $\forall
\zeta,\eta\in Y$. In particular item {\it d)} holds for~$b$.
\item For all $\zeta,\eta$:
\begin{eqnarray*} b(\zeta,i\eta)&=& B(\zeta,i\eta)-i B(\zeta,-\eta)=B(\zeta,i\eta)+iB(\zeta,\eta)\\
&=&i\left[ B(\zeta,\eta) -i B(\zeta,i\eta)
\right] = ib(\zeta,\eta)
\end{eqnarray*}
\item For all $\zeta,\eta$:
\begin{eqnarray*} b(\eta,\zeta)&=& B(\eta,\zeta)-iB(\eta,i\zeta) = B(\zeta,\eta)-i B(i\zeta,\eta) \\
&=& B(\zeta,\eta)-i B(-\zeta,i\eta)=B(\zeta,\eta) + i B(\zeta,i\eta) =
\overline{b(\zeta,\eta)}.
\end{eqnarray*}
\item For all $t,s\in\R$ and $\zeta,\eta\in Y$:
\begin{eqnarray*}
 b(\zeta, (t+is)\eta)&=& b(\zeta,t\eta) + b(\zeta,is\eta)\\&=&tb(\zeta,\eta) + i s
b(\zeta,\eta)=(t+is)b(\zeta,\eta).
\end{eqnarray*}Together with the above relations this also implies $b(z\zeta)=|z|^2b(\zeta)$,
$\forall z\in\C$, that is, item \rm {e)} holds for $b$.
\end{itemize}

Finally, for $\zeta\in Y$ and $z\in\C$, $b(z\zeta)= b(z\zeta,z\zeta) = B(z\zeta,z\zeta)- i
B(z\zeta,iz\zeta)$, and by selecting $z=i$ together with item {\it d}), it follows that
\begin{eqnarray*} b(i\zeta)&=&b(\zeta),\\ B(i\zeta,i\zeta)-i B(i\zeta,-\zeta)  &=& B(\zeta,\zeta) - i
B(\zeta,i\zeta), \\
 B(i\zeta,i\zeta)+i B(\zeta,i\zeta) &=& B(\zeta,\zeta) -i B(\zeta,i\zeta),
\end{eqnarray*}so that $2i B(\zeta,i\zeta) = B(\zeta,\zeta)-B(i\zeta,i\zeta)$.  Since
$B(\cdot)$ is real, it is found that $B(\zeta,i\zeta)=0$ and so
\[ b(\zeta) = B(\zeta,\zeta)-i B(\zeta,i\zeta) = B(\zeta,\zeta)=F(\zeta).
\]This implies that $b$ satisfies {\it a),b),c),f)} since $F$ does, and  the proof of the proposition
is completed. 
\end{proof} 

The above results allow us to prove the following complex versions of important results previously proven for real Hilbert spaces and presented in Dal Maso's book. By taking into account the above propositions, the proofs are simple variations of their counterparts in the real case, and so they will be omitted. 

\begin{Proposition}\label{propCharacQuadra} Let  $T:\dom T\to\hil$ be a positive self-adjoint operator,
$\overline{\dom T}= \hil_0$, and $b^T:\hil\to\overline\R$ the quadratic form generated by~$T$. Then
\begin{eqnarray*} b^T(\zeta)&=&\sup_{\eta\in\dom T} \left[ \la T\eta,\zeta\ra +\la\zeta,T\eta\ra - \la
T\eta,\eta\ra
\right] 
\\ &=& \sup_{\eta\in\dom T} \left[ b^T(\eta)+ \la T\eta,\zeta\ra +\la\zeta,T\eta\ra - 2\la
T\eta,\eta\ra
\right],
\end{eqnarray*}for all $\zeta\in \hil_0$ and $b^T(\zeta)=\infty$ if $\zeta\in
\hil\setminus\hil_0$.
\end{Proposition}

\begin{Remark} The main difference in the proof of Proposition~\ref{propCharacQuadra} with respect to
the real case (see Theorem~12.21 in~\cite{DalMaso}) is the following. For
$\zeta\in\hil_0$ denote 
\[ F(\zeta) = \sup_{\eta\in\dom T} \left[ \la T\eta,\zeta\ra +\la\zeta,T\eta\ra - \la T\eta,\eta\ra
\right] 
\]and one needs to check that $F$ is a quadratic form, and a way of doing this is to employ our Proposition~\ref{ComplexQF}, in particular to check
\begin{eqnarray*} F(i\zeta) &=&  \sup_{\eta\in\dom T} \left[ \la T\eta,i\zeta\ra +\la i\zeta,T\eta\ra
- \la T\eta,\eta\ra \right] 
\\ &=&  \sup_{(-i\eta)\in\dom T} \left[ \la T(-i\eta),\zeta\ra +\la\zeta,T(-i\eta)\ra - \la
T(-i\eta),(-i\eta)\ra
\right] 
\\ &=& F(\zeta).
\end{eqnarray*}
\end{Remark}

By using the above proposition one gets the following complex version of Theorem~13.5 in
\cite{DalMaso}: 
\begin{Theorem}\label{GammaWeakRW} Let $b_\varepsilon,b\ge \beta>0$ be sesqui\-lin\-e\-ar forms on the complex~$\hil$ and $T_\varepsilon,T\ge\beta\Id$ the corresponding associated self-adjoint operators, and let $\overline{\dom T}=\hil_0\subset\hil$. Let $P_0$ denote the orthogonal projection onto~$\hil_0$. Then the  following  statements are equivalent: 
\begin{itemize}
\item[i)] $b_\varepsilon \wgconv b$.
\item[ii)] $R_0(T_\varepsilon)$ converges weakly to $R_0(T)P_0$.
\end{itemize} 
\end{Theorem}

With such ``complex'' tools at hand, the proof of Theorem~\ref{mainTheorGamma} follows the same steps of the proof of its real counterpart, i.e.,  Theorem~13.6 and Corollary~13.7 in~\cite{DalMaso}. 

Now we state sufficient conditions to obtain norm resolvent convergence of  operators from $\Gamma$-convergence. The following theorem was proven in~\cite{CRO}, the proof for complex Hilbert spaces is similar and doesn't require further comments.

\begin{Proposition}\label{GammaNorm} Let~$\hil$ be a real or complex Hilbert space, $b_\varepsilon,b\ge \beta>-\infty$ be closed sesqui\-lin\-e\-ar forms and $T_\varepsilon,T\ge\beta\Id$ the corresponding associated self-adjoint operators, and let $\overline{\dom T}=\hil_0\subset\hil$. Assume that the following three conditions hold:
\begin{itemize}
\item[a)] $b_\varepsilon\sgconv b$ and $b_\varepsilon\wgconv b$. 
\item[b)] The resolvent operator $R_{-\lambda}(T)$ is compact in $\hil_0$ for some real number
$\lambda>|\beta|$.
\item[c)] There exists a Hilbert space $\mathcal K$,  compactly embedded in~$\hil$, so that if $(\psi_\varepsilon)$ is bounded in~$\hil$ and the sequence $(b_\varepsilon(\psi_\varepsilon))$ is also bounded, then $(\psi_\varepsilon)$ is a bounded subset of~$\mathcal K$.
\end{itemize} Then, $T_\varepsilon$ converges in norm resolvent sense to~$T$ in $\hil_0$ as $\varepsilon\to0$.
\end{Proposition}

\subsection{Norm convergence of quadratic forms}
In some situations the convergence of quadratic forms may also imply the norm resolvent convergence of the corresponding operators.  This subject was discussed in~\cite{AAV} for real Hilbert spaces. In the following, we (re)state and prove a complex version  of a result in~\cite{AAV} which will be useful in this work; we also correct an imprecision in the previous proof.
   
\begin{Theorem}\label{teocorrecao}
Let $(b_\var)_\var$, $(m_\var)_\var$ be two sequences of positive and closed sesquilinear quadratic forms in a complex Hilbert space
${\cal H}$ with $\dom b_\var = \dom m_\var = {\cal D}$,
for all $\var >  0$, and $B_\var$, $M_\var$ the self-adjoint operators associated with $(b_\var)_\var$
and $(m_\var)_\var$, respectively. Suppose that there is $\lambda >0$ so that $b_\var, m_\var \geq \lambda$, for
all $\var > 0$, and
\begin{equation}\label{quadrcorrecao}
|b_\var(\psi) - m_\var(\psi)| \leq q(\var) \, m_\var(\psi), \qquad \forall \psi \in {\cal D},
\end{equation}
with $q(\var) \too 0$ as $\var \too 0$. Then, there exists $C > 0$ so that, for $\var > 0$ small enough,
$$\|  B_\var^{-1} - M_\var^{-1} \|  \leq C \, q(\var).$$
\end{Theorem}
\begin{proof}
Let $b_\var(u, \tilde{u})$ and $r_\var(u, \tilde{u})$ the sesquilinear forms associated with
$b_\var(u)$ and $m_\var(u)$, respectively. Recall the polarization identity
$$b_\var(u, \tilde{u}) = \frac{1}{4} \left[ b_\var(u+\tilde{u}) - b_\var(u -\tilde{u}) - i\,
b_\var(u + i \, \tilde{u}) + i \, b_\var(u -i \, \tilde{u})\right],$$
which will be used ahead.

Note that condition~\eqref{quadrcorrecao} implies 
$$(1 - q(\var)) m_\var(\psi) \leq b_\var(\psi) \leq (1 + q(\var)) m_\var(\psi), \qquad \forall \psi \in{\cal D}.$$
As $q(\var) \too 0$, there exist $\var_0 > 0$ and a number $C_1 > 0$ so that
$m_\var(\psi) \leq C_1 \, b_\var(\psi)$, for all~$ \var < \var_0$ and~$ \psi \in {\cal D}$.

For $u, \tilde{u} \in {\cal D}$, one has
\begin{eqnarray*}
& &
\left| \langle B_\var^{1/2} u, B_\var^{1/2} \tilde{u}    \rangle -  \langle M_\var^{1/2} u, M_\var^{1/2} \tilde{u}    \rangle \right| 
 = 
\left| b_\var(u, \tilde{u}) - m_\var(u, \tilde{u})  \right| \\ 
& = &
(1/4) \, \left| b_\var(u+\tilde{u}) - m_\var(u+\tilde{u}) - b_\var(u-\tilde{u}) + m_\var(u-\tilde{u}) \right.\\
& - & \left.
i \, b_\var(u + i \, \tilde{u}) + i \, m_\var(u + i \, \tilde{u}) 
+ i \, b_\var(u - i \, \tilde{u}) - i \, m_\var(u - i \, \tilde{u}) \right| \\
& \leq &
(1/4) \, q(\var)  \left[
m_\var(u+\tilde{u}) + m_\var(u-\tilde{u}) +  m_\var(u+ i \,\tilde{u}) + m_\var(u-i \,\tilde{u})\right] \\
& = &
q(\var) 
\left[ m_\var(u) + m_\var(\tilde{u}) \right] \\
& \leq &
q(\var) 
\left[ C_1 b_\var(u) + m_\var(\tilde{u}) \right].
\end{eqnarray*}
Taking $u = B_\var^{-1} g$ and $\tilde{u} = M_\var^{-1} \tilde{g}$, with $g, \tilde{g} \in \LL^2({\cal H})$, one has
\begin{eqnarray*}
\left|\langle (B_\var^{-1}-M_\var^{-1}) g, \tilde{g} \rangle \right| & \leq &
q(\var) \left[ C_1\, \langle B_\var^{-1} g,  g \rangle +  \langle M_\var^{-1} \tilde{g},  \tilde{g} \rangle  \right] \\
& \leq &
q(\var) \left[ C_1 \| B_\var^{-1} \| \| g \|^2 + \| M_\var^{-1}\| \|\tilde{g}\|^2 \right].
\end{eqnarray*}
Thus,
$$\| B_\var^{-1} - M_\var^{-1}\| = \sup_{\|g\|=1} \left|\langle (B_\var^{-1}-M_\var^{-1}) g, g \rangle \right| \leq 
q(\var) \left[ C_1 \| B_\var^{-1} \|  + \| M_\var^{-1}\|  \right]
\leq  C q(\var),$$
for some~$C> 0$.
\end{proof}

\section{The model}\label{themodel}

\subsubsection*{ The geometry of the domain}
Let~$S$ be a circle of length $l > 0$ and  $r:S \rightarrow \mathbb R^3$ a closed and simple curve of class~$\CC^3$ in~$\mathbb R^3$ parameterized by its arc length parameter~$s$.  Just as in~\cite{BMT, OCA}, we assume that~$r(s)$ is endowed with the Frenet trihedron consisting  of orthogonal unit vectors  $\{T(s), N(s), B(s)\}$ satisfying the system of Frenet equations (as usual, we take the tangent, normal and binormal vectors). We denote by~$k(s)$ and~$\tau(s)$ the curvature and torsion, respectively,  of the curve~$r$ at the position~$r(s)$; due to continuity, such functions are bounded.

Let~$Q$ be a nonempty open, bounded, connected and simply connected  subset of~$\mathbb R^2$, and with a smooth boundary. The set
$$\Omega = \{x\in \mathbb R^3 : x = r(s) + y_2 N(s) + y_3 B(s), s \in S , y = (y_2,y_3)\in Q \}$$
is obtained by putting the region~$Q$ along the curve~$r(s)$.
In each point~$r(s)$ one also allows a rotation angle~$\alpha(s)$ of the cross-section~$Q$, and such rotation function is  supposed to be of class~$\CC^2$. Thus, the new region is given by
\[
\Omega^\alpha = \{ x \in \mathbb R^3 : x = r(s) + y_2 N_\alpha(s) + y_3 B_\alpha(s), s \in S , y = (y_2,y_3) \in Q \},
\]
where
\begin{eqnarray*}
N_\alpha(s) & := & \cos \alpha(s) N(s) + \sin \alpha(s) B(s),\\
B_\alpha(s) & := & -\sin \alpha(s) N(s) + \cos \alpha(s) B(s).
\end{eqnarray*}

Now, we add a small parameter $\varepsilon > 0$ to obtain the sequence of regions
$$\Omega^\alpha_\varepsilon = \{ x \in \mathbb R^3 : x = r(s) + \varepsilon y_2 N_\alpha(s) + \varepsilon y_3 B_\alpha(s), 
s \in S, y = (y_2,y_3) \in Q \},$$
which is ``squeezed''  to the curve~$r(s)$ as $\varepsilon \rightarrow 0$.

\subsubsection*{Quadratic forms}
As  mentioned in the introduction of this work, we consider the vector magnetic potential ${\bf A} = (A_1,A_2,A_3)$, where $A_j: \Omega \rightarrow \mathbb R$, $j=1,2,3$, are real functions, and the family of self-adjoint magnetic Schr\"odinger  operators
\begin{eqnarray*}
H_\varepsilon^\alpha \psi & = & \left(-i \partial_x - {\bf A}\right)^2 \psi \\
& =  & 
\left(-i \partial_{x_1} - A_1\right)^2 \psi + (-i \partial_{x_2} - A_2)^2 \psi + \left(-i\partial_{x_3} - A_3\right)^2 \psi,
\end{eqnarray*}
 $\dom H_\varepsilon^\alpha = \hil^2(\Omega^\alpha_\varepsilon)\cap \hil^1_0(\Omega^\alpha_\varepsilon)$. 
$\partial_{x_j}$ denotes the partial derivative with respect to the coordinate $x_j$, and so on.

We suppose that the vector field ${\bf A}$  is continuous on the reference curve~$S$, 
\begin{equation}\label{condA}
 A_j \in   W^{1,\infty}(\Omega),\quad j=1,2,3, 
\end{equation} and both restrictions $s\mapsto A_2(r(s),0,0), s\mapsto A_3(r(s),0,0)$ belong to $W^{2,\infty}(S)$  (the latter condition is due to the gauge transform~\eqref{nullpotentialcurve}). These relatively weak regularity of the magnetic potential is possible thanks to the technique of $\Gamma$-convergence employed here.

The family of quadratic forms associated with the operators~$H_\varepsilon^\alpha$ is given by
\begin{eqnarray*}\label{formasquadraticasaaa}
b_\varepsilon^\alpha(\psi) & := & 
\int_{\Omega^\alpha_\varepsilon}\left|(-i\partial_x-\bf{A})\psi\right|^2 {\mathrm d}x \\
& = & \int_{\Omega^\alpha_\varepsilon}\left(\left|(-i\partial_{x_1}-A_1)\psi\right|^2+\left|(-i\partial_{x_2}-A_2)\psi\right|^2+\left|(-i\partial_{x_3}-A_3)\psi\right|^2\right) {\mathrm d}x,
\end{eqnarray*}
with $\dom b_\varepsilon^\alpha = \hil^1_0(\Omega^\alpha_\varepsilon)$.

\subsubsection*{Change of variables}
Now we are going to perform a change of variables 
so that the integration region in~$b_\var^\alpha$, and consequently their domains, don't depend on the parameter $\varepsilon>0$. The change ahead is usual  and details will  be omitted; see~\cite{BMT, CRO, OCA}.

For each $\var > 0$ consider the function
$$
\begin{array}{lccl}
f^\alpha_\varepsilon: &  S \times Q    &    \rightarrow    &    \Omega^\alpha_\varepsilon\\
                    & (s,y_2,y_3)  &    \mapsto   &   r(s)+\varepsilon y_2N_\alpha(s)+\varepsilon y_3B_\alpha(s),
\end{array}
$$
and the unitary operator
$$
\begin{array}{cccc}
U^\alpha_\varepsilon:  &   \LL^2(\Omega^\alpha_\varepsilon)  &   \rightarrow  &   \LL^2(S \times Q,\beta_\varepsilon)\\
                       &             \psi                  &     \mapsto& \varepsilon \, \psi\circ f^\alpha_\varepsilon
\end{array}$$
where
$\beta_\var(s,y) := 1 - \var k(s) \langle z_\alpha, y\rangle$ and $z_\alpha: = (\cos \alpha, \sin \alpha)$
($\beta_\var$ comes from the Riemannian metric  defined by $f_\var^\alpha$, which is  a global diffeomorphism for~$\var> 0$ small enough).

We denote 
$$\left(\begin{array}{ccc}
\hat{A}_1^\var(s,y) \\
\hat{A}_2^\var(s,y) \\
\hat{A}_3^\var(s,y)
\end{array}\right) : = \left(\begin{array}{ccc}
T(s) \\
N(s) \\
B(s)
\end{array}\right)
\left(\begin{array}{ccc}
\left(A_1 \circ f_\var^\alpha\right)(s,y)  \\
\left(A_2 \circ f_\var^\alpha\right)(s,y) \\
\left(A_3 \circ f_\var^\alpha\right)(s,y)
\end{array}\right), 
$$
\begin{eqnarray*}
\tilde{A}_2^\varepsilon(s, y)  
& := &  
\varepsilon\cos\alpha(s)\hat{A}_2^\var(s,y) +\varepsilon\, \sin \alpha(s)\hat{A}_3^\var(s,y),\\
\tilde{A}_3^\varepsilon(s, y)
& := & -\varepsilon\,\sin \alpha(s)\hat{A}_2^\var(s,y)+\varepsilon\cos\alpha(s)\hat{A}_3^\var(s,y),
\end{eqnarray*}
and, for $\varphi \in \hil_0^1(S \times Q)$,
\begin{eqnarray*}
\tilde{\partial}_{y_2}\, \varphi & : = & (-i\partial_{y_2}- \tilde{A}_2^\varepsilon(s, y)) \, \varphi, \\
\tilde{\partial}_{y_3}\, \varphi & : = & (-i\partial_{y_3}- \tilde{A}_3^\varepsilon(s, y)) \, \varphi.
\end{eqnarray*}

From now on we  consider the sequence $b_\var^\alpha+c$, with $c > \| k^2/4\|_\infty$ (the reason for this choice will be clear ahead).
Some calculations show that the quadratic form $b_\var^\alpha+c$, 
after of the change given by $U_\var^\alpha$,
can be written as
\begin{eqnarray*}
\tilde{b}^{\alpha,c}_\varepsilon(\varphi)
& := &
\int_{S \times Q}
\frac{1}{\beta_\varepsilon}
\left| -i \frac{\partial \varphi}{\partial s} - \beta_\varepsilon \hat{A}_1^\var(s,y) \,\varphi -
\la \nabla_y \varphi,  Ry \ra (\tau+\alpha') \right|^2 {\mathrm d}s{\mathrm d}y\\
& + &
\int_{S \times Q}\left(
\frac{\beta_\varepsilon}{\varepsilon^2} | \tilde{\partial}_y \varphi|^2 + c \, \beta_\varepsilon |\varphi|^2 \right) {\mathrm d}s{\mathrm d}y,
\end{eqnarray*}
$\dom \tilde{b}^{\alpha,c}_\varepsilon=\hil^1_0(S \times Q)$, where
$\nabla_y := (-i \partial_{y_2}, -i \partial_{y_3})$,
$\tilde{\partial}_y := (\tilde{\partial}_{y_2}, \tilde{\partial}_{y_3})$ 
and
$R$ is the rotation matrix
$\left(
\begin{array}{cc}
	0&1\\ -1&0
\end{array}
\right).$

The domain of each  $\tilde{b}_\varepsilon^{\alpha, c}$  is the subspace $\hil^1_0(S \times Q)$ of the Hilbert space
$\LL^2(S \times Q,\beta_\varepsilon(s,y))$.
However, it is convenient to work in $\LL^2(S \times Q)$, that is, with the usual Lebesgue measure. Thus, we
consider the isometry
\begin{equation}\label{mudançavariavel}
\begin{array}{cccc}
V_\var^\alpha: & \LL^2(S \times Q)&\rightarrow&\LL^2(S \times Q,\beta_\epsilon)\\
                      &    v    &     \mapsto   &    \beta^{-1/2}_\varepsilon \, v
\end{array}.
\end{equation}

Now, for $v \in \hil_0^1(S \times Q)$, we denote 
$$\tilde{\partial}_s \, v : =   (-i\partial_s- \tilde{A}_1^\varepsilon(s, y)) \, v,$$
where
$$\tilde{A}_1^\varepsilon(s, y) 
:=  
\beta_\varepsilon(s,y) \hat{A}_1^\var(s,y) - \frac{i}{2\beta_\varepsilon} \frac{\partial \beta_\varepsilon}{\partial s}
- \frac{1}{2 \beta_\varepsilon} \la \nabla_y \beta_\varepsilon, R y \ra (\tau+\alpha')(s).$$

Applying the change of variables (\ref{mudançavariavel}) to the quadratic form $\tilde{b}_\varepsilon^{\alpha,c}$, it is found that
\begin{eqnarray*}
\tilde{g}^{\alpha,c}_\varepsilon(v) & = & 
\int_{S \times Q}  (1/\beta^2_\varepsilon) 
\left|  
\tilde{\partial}_s v - \la \nabla_y \psi , R  y \ra (\tau+\alpha')(s)\right|^2\\
& + & \frac{1}{\varepsilon^2} \int_{S \times Q}  |\tilde{\partial}_y v|^2  {\mathrm d}s{\mathrm d}y
- \int_{S \times Q}  (1/\beta^2_\varepsilon) \frac{k^2(s)}{4}|v|^2  {\mathrm d}s{\mathrm d}y \\
& + &
c \int_{S \times Q}  |v|^2  {\mathrm d}s{\mathrm d}y,
\end{eqnarray*}
where $\dom \tilde{g}^{\alpha,c}_\varepsilon = \hil^1_0(S \times Q)$
is now a subspace of $\LL^2(S \times Q)$.

Besides allowing us to working in the Hilbert space  $\LL^2(S \times Q)$ with the usual  measure, the unitary transformation~(\ref{mudançavariavel}) makes the curvature appears in the expression of the quadratic form (note the penultimate integral in the definition of $\tilde{g}_\varepsilon^{\alpha,c}$).
Now it is clear the role played by  the constant~$c > 0$: for $\var > 0$ small enough,
the quadratic forms $\tilde{g}_\varepsilon^{\alpha, c}(v)$ become positive.

\subsubsection*{Renormalization of $\tilde{g}^{\alpha, c}_\varepsilon$}
When the sequence of tubes is ``squeezed'' to the curve~$r(s)$, there are divergent eigenvalues  due to the factor 
\[
\frac1{\varepsilon^2} \int_{S \times Q} |\tilde{\partial}_y\psi|^2 \, {\mathrm d}s{\mathrm d}y,
\]
which is directly related with the magnetic Laplacian restricted to the cross section~$Q$. We renormalize this divergence by subtracting  
\begin{equation}\label{eqRenorSemMag}
\frac{\lambda_0}{\varepsilon^2}\int_{S \times Q} |\psi|^2 \, {\mathrm d}s{\mathrm d}y
\end{equation}  from $\tilde{g}_\varepsilon^{\alpha,c}(\psi)$. Recall that~$\lambda_0$ is the first eigenvalue of the Dirichlet Laplacian in~$Q$ (see (\ref{crosssectioneigenvalue}) in the introduction of this work). Thus, we pass to consider the  sequence of renormalized quadratic forms
\begin{eqnarray*}
g_\varepsilon^{\alpha, c}(v)   & : = & \tilde{g}_\varepsilon^{\alpha, c}(v) - \int_{S \times Q} \frac{\lambda_0}{\varepsilon^2} |v|^2  {\mathrm d}s{\mathrm d}y \\
& = &  
\int_{S \times Q}  (1/\beta^2_\varepsilon) 
\left|  
\tilde{\partial}_s v - \la \nabla_y v,  R \, y \ra(\tau+\alpha')(s)  \right|^2  {\mathrm d}s{\mathrm d}y \\
& + & \frac{1}{\varepsilon^2} \int_{S \times Q} \left( |\tilde{\partial}_y v|^2 - \lambda_0 |v|^2 \right)  {\mathrm d}s{\mathrm d}y
- \int_{S \times Q}  (1/\beta^2_\varepsilon) \frac{k^2(s)}{4}|v|^2  {\mathrm d}s{\mathrm d}y \\
& + &
c \int_{S \times Q}  |v|^2  {\mathrm d}s{\mathrm d}y,
\end{eqnarray*}
with $\dom g^{\alpha, c}_\varepsilon = \hil^1_0(S \times Q)$.

\subsubsection*{Gauge transform}
For the magnetic field ${\bf A}$ we are going to suppose, without loss of generality,
that
\begin{equation}\label{nullpotentialcurve}
\la N(s) , {\bf A}(r(s)) \ra =  \la B(s) , {\bf A}(r(s)) \ra =0, \quad \forall s \in S.
\end{equation} 
In fact,
let $\tilde{{\bf A}}_\var := (\tilde{A}_1^\var, \tilde{A}_2^\var, \tilde{A}_3^\var)$,
by using the gauge transform 
$$\tilde{{\bf A}}_\var \mapsto \tilde{{\bf A}}_\var - \nabla \Phi_\var,$$
with $\Phi_\var(s, y)= y_2 \tilde{A}_2^\var (s,0)  + y_3 \tilde{A}_3^\var (s,0)$,
we can suppose
$\tilde{A}_2^\var(s,0)= \tilde{A}_3^\var(s,0) = 0$, for all $s \in S$, which implies
the condition
(\ref{nullpotentialcurve}).
Due the periodicity (recall that~$r(s)$ is a closed curve),  usually the vector potential in one-dimensional effective operators can not be gauged away, as it happens  in case of unbounded tubes~\cite{DKRAY}.

Our study, in the next sections, will be conducted with the quadratic forms~$g_\varepsilon^{\alpha, c}(v)$; thus, for simplicity,
we shall omit the indices~$\alpha$ and~$c$ from the notations. For example, $g^{\alpha,c}_\varepsilon$  will be simply denoted  by~$g_\varepsilon$.

\section{The main theorem}\label{maintheorem}

In this section we present our application of the complex $\Gamma$-convergence discussed in Section~\ref{secIntro}. Recall that
$$g_\epsilon(v) =  + \infty \hspace{0.5cm} \hbox{for} \hspace{0.5cm} v\in \LL^2(S \times Q)\setminus \hil^1_0(S \times Q),$$
put
\begin{equation}\label{constC}
C(Q):= \int_Q \left| \la \nabla_y  u_0, Ry \ra \right|^2\,{\mathrm d}y =  \int_Q \left| (y_2\partial_{y_3} - y_3\partial_{y_2})  u_0 \right|^2\,{\mathrm d}y,
\end{equation}
and recall that~$u_0$ is a normalized eigenfunction  corresponding to the first eigenvalue~$\lambda_0$ of the Dirichlet Laplacian in~$Q$. Note that~$C(Q)$ depends only on the cross section~$Q$. 

Consider  the one-dimensional quadratic form 
\begin{eqnarray*}
g_0(w) & := & \int_{S}\left| \left(  -i  \partial_s - \la {\bf A}(r(s)), T(s) \ra   \right) w(s) \right|^2 {\mathrm d}s\\
& + &
\int_{S } \left[C(Q)(\tau+\alpha')^2(s)-\frac{k^2(s)}{4}+c\right]| w(s)|^2 {\mathrm d}s,
\end{eqnarray*}
with $\dom g_0 = \hil^1(S)$. From  $g_0(w)$ we define a quadratic form on $\LL^2(S \times Q)$:
\begin{eqnarray}\label{formaquadraticalimite}
g(v)=\left\{
\begin{array}{cc}
g_0(w) & \hbox{if} \hspace{0.5cm} v = w u_0 \hspace{0.5cm} \hbox{with} \hspace{0.5cm} w \in \dom g_0 \\
+ \infty &  \hbox{otherwise} 
\end{array}.
\right.
\end{eqnarray}

We denote by $G_\varepsilon$ and ~$G$  the respective self-adjoint operators associated with the quadratic forms $g_\varepsilon$ and~$g$.  Our application is the following theorem.

\begin{Theorem}\label{teoremaprincipal}
Under condition (\ref{nullpotentialcurve}), the sequence of operators  $G_\varepsilon$ converges to~$G$ in the norm convergence sense as~$\varepsilon \to 0$. More exactly,
$$\left \| (G_\var - i \, {\bf I})^{-1} - (G-i \, {\bf I})^{-1} \right\| \too 0$$
as $\var \too 0$.
\end{Theorem}

Through the unitary transformation $v(x,y)=w(s)u_0(y)\mapsto w(s)$, $G$ can be identified with the one-dimensional operator
$$(G_0 w)(s) := (-i \partial_s - \la {\bf A}(r(s)),  T(s)\ra)^2 w(s) + \left[C(Q) (\tau + \alpha')^2(s)
-\frac{k^2(s)}{4} + c \right]w(s),$$
with $\dom G_0 = \hil^2(S ).$
Due to this identification, we say that there was a  ``reduction of dimension'' 
in the limit~$\varepsilon \to 0$.
We can also note the presence of a  potential in~$G_0$ that came from the original magnetic potential, as well as a term that depends on geometric  effects of the original region, as habitually is the case in such kind of problems without magnetic fields.

To prove  Theorem~\ref{teoremaprincipal} we are going to apply  Proposition~\ref{GammaNorm} of  Section~\ref{secIntro}.
The first step is to show that
the sequence  $g_\varepsilon$ strongly $\Gamma$-converges to~$g$.
Observe that in our case the strong and weak $\Gamma$-convergences are equivalent because the region $S \times Q$ is bounded.

Some important properties that we will use ahead are as follows:

\begin{itemize}
\item[\bf{(1)}] 
For all function $v \in \hil_0^1(S \times Q)$, we have
\begin{equation}\label{firsteigendiclet}
\int_{Q} \left( |\nabla_y v(s, y)|^2 - \lambda_0 |v(s, y)|^2 \right) \, {\mathrm d}y 
\geq 0, \qquad \hbox{a.e.[s]},
\end{equation}
and
\begin{equation}\label{desigualdadediamagnetica}
\int_{Q} \left( |\tilde{\partial}_y v(s, y)|^2 - \lambda_0 |v(s, y)|^2 \right) \, {\mathrm d}y 
\geq 0, \qquad \hbox{a.e.[s]}.
\end{equation}
The first inequality follows from the definition of $\lambda_0$ and the second one is a consequence of (\ref{firsteigendiclet}) combined with the Diamagnetic Inequality~\cite{Diamag}.

\item[\bf{(2)}] The Dirichlet condition on the boundary $\partial Q$ implies $\displaystyle\int_Q \nabla_y|u_0|^2\,{\mathrm d}y=0$, which, by its turn, implies that
$$
\int_Q \la u_0\nabla_y u_0,  Ry \ra\,{\mathrm d}y=0. 
$$
\end{itemize}

To show the strong  $\Gamma$-convergence of the sequence $g_\varepsilon$  we will make use of some lemmas.
\begin{Lemma}\label{lema01teorema}
If $v_\varepsilon \rightharpoonup v$ in $\LL^2(S \times Q)$ and 
$(g_\varepsilon(v_\varepsilon))_\varepsilon$ is
a bounded sequence in $\LL^2(S \times Q)$, then $(\partial_s v_\varepsilon)_\varepsilon$ and $(\nabla_y v_\varepsilon)_\varepsilon$ 
are bounded sequences in $\LL^2(S \times Q)$. 
Furthermore, $\partial_s v_\varepsilon \rightharpoonup \partial_s v,\nabla_y v_\varepsilon \rightharpoonup \nabla_y v$ 
in $\LL^2(S \times Q)$ and $v\in \hil^1_0(S \times Q)$.
\end{Lemma}
\begin{proof}
As $(g_\varepsilon(v_\varepsilon))_\varepsilon$ is a bounded sequence, there exists a number~$D>0$ so that
\begin{eqnarray*}\label{e2}
& &
\limsup_{\varepsilon \rightarrow 0} \int_{S \times Q} \frac{1}{\beta_\varepsilon^2}  
\left| \tilde{\partial}_s v_\varepsilon  -
\la \nabla_y v_\varepsilon, Ry \ra (\tau+\alpha')(s) \right|^2  {\mathrm d}s\,{\mathrm d}y \\
&\leq &
\limsup_{\varepsilon\rightarrow0} g_\varepsilon(v_\varepsilon) \leq  D.
\end{eqnarray*}

Now, since $(v_\varepsilon)_\varepsilon$ is also a bounded sequence, we have
\begin{eqnarray*}\label{e3}
& &
\limsup_{\varepsilon\rightarrow0}  \int_{S \times Q} 
\left|\tilde{\partial}_y v_\varepsilon\right|^2 {\mathrm d}s{\mathrm d}y  \\
& = &
\limsup_{\varepsilon\rightarrow0}  
\left(\int_{S \times Q} \left(\left|\tilde{\partial}_y v_\varepsilon\right|^2 - \lambda_0
|v_\varepsilon|^2\right) {\mathrm d}s{\mathrm d}y +  \int_{S \times Q} \lambda_0 |v_\varepsilon|^2  \, {\mathrm d}s{\mathrm d}y\right) \\
&\leq & 
\limsup_{\varepsilon\rightarrow0} \, D \varepsilon^2 + \limsup_{\varepsilon\rightarrow0}  \int_{S \times Q}
\lambda_0 |v_\varepsilon|^2 {\mathrm d}s\, {\mathrm d}y< \infty.
\end{eqnarray*}

We have shown above that  $((-i\partial_{y_2}-\tilde{A}_2^\var)v_\varepsilon)_\varepsilon$ and $((-i\partial_{y_3}-\tilde{A}_3^\var)v_\varepsilon)_\varepsilon$ are bounded sequences in $\LL^2(S \times Q)$.  Since $(\tilde{A}_2^\var v_\varepsilon)$ and $(\tilde{A}_3^\var v_\varepsilon)$ are also bounded sequences in $\LL^2(S \times Q)$, we conclude that $(\nabla_y v_\varepsilon)_\varepsilon$ is a bounded sequence in $\LL^2(S \times Q)$. Similarly, $(\partial_s v_\varepsilon)_\varepsilon$ is a bounded sequence in $\LL^2(S \times Q)$.  Therefore, $(v_\varepsilon)_\varepsilon$ is a bounded sequence in $\hil_0^1(S \times Q)$. Thus, there exist $\phi \in \hil_0^1(S \times Q)$ and a subsequence of $(v_\varepsilon)_\varepsilon$, also denoted by $(v_\varepsilon)_\varepsilon$, so that $v_\var \rightharpoonup \phi$ in $\hil_0^1(S \times Q)$ (recall that Hilbert spaces are reflexive). As $v_\varepsilon \rightharpoonup v$ in $\LL^2(S \times Q)$, it follows that~$v=\phi$, $\partial_s v_\varepsilon \rightharpoonup \partial_s v,\nabla_y v_\varepsilon \rightharpoonup \nabla_y v$ in $\LL^2(S \times Q)$ and $v\in \hil^1_0(S \times Q)$.
\end{proof}

\begin{Lemma}\label{lema02teorema}
Let $v_\varepsilon\rightarrow v$ be in $\LL^2(S \times Q)$ so that  there exists $\displaystyle\lim_{\varepsilon\rightarrow0}g_\varepsilon(v_\varepsilon)<\infty$. 
Then, we can write $v(s,y)=w(s)u_0(y)$ with $w\in \hil^1(S )$.
\end{Lemma}
\begin{proof}
In fact, by previous lemma, $\nabla_y v_\varepsilon\rightharpoonup\nabla_y v$ weakly in $\LL^2(S \times Q)$.  Observe that also $\tilde{A}_2^\var v_\varepsilon\rightharpoonup0$ and $\tilde{A}_3^\var v_\varepsilon\rightharpoonup0$ weakly. Thus, $\tilde{\partial}_y v_\varepsilon\rightharpoonup \nabla_y v$ weakly in $\LL^2(S \times Q)$.

From the strong convergence of~$(v_\varepsilon)_\varepsilon$  we have
\begin{eqnarray*}
\int_{S \times Q}
\left|\nabla_y v\right|^2  {\mathrm d}s{\mathrm d}y 
& \leq & 
\liminf_{\varepsilon \rightarrow 0}\int_{S \times Q} \left|\tilde{\partial}_y v_\varepsilon\right|^2  {\mathrm d}s{\mathrm d}y \\
& \leq & 
\limsup_{\varepsilon \rightarrow 0} \int_{S \times Q} \lambda_0 |v_\varepsilon|^2  {\mathrm d}s{\mathrm d}y\\
& = &
\lambda_0 \int_{S \times Q}|v|^2  {\mathrm d}s{\mathrm d}y.
\end{eqnarray*}

This fact, combined with (\ref{firsteigendiclet}) above, tell us that
$$
\int_{S \times Q}\left(\left|\nabla_y v\right|^2-\lambda_0|v|^2\right)  {\mathrm d}s{\mathrm d}y=0. $$
 
Consider $\displaystyle f(s):=\int_Q\left(\left|\nabla_y v(s,y)\right|^2 - \lambda_0|v(s,y)|^2\right) {\mathrm d}y$.  Since $f(s) \geq 0$, the inequality above implies that $f=0$ a.e.[s]. Therefore, $v(s,y)$ is an eigenfunction associated with~$\lambda_0$.  As~$\lambda_0$ is a simple eigenvalue,  $v(s,y)$ is proportional to $u_0$.  Thus, we can write  $v(s,y)=w(s)u_0(y)$ with $w\in \hil^1(S)$ (since $v\in \hil^1_0(S \times Q)$).
\end{proof}

\vspace{0.5cm}
\noindent
{\bf Proof of Theorem~\ref{teoremaprincipal}:}
Let $v \in \LL^2(S \times Q)$ and 
$v_\varepsilon \too v$ in $\LL^2(S  \times Q)$. 
We are going to show that
$\lim_{\varepsilon \too 0} g_\varepsilon(v_\varepsilon) \geq g(v)$.
If $\displaystyle\liminf_{\varepsilon \rightarrow 0}g_\varepsilon(v_\varepsilon)=\infty$,  then $\displaystyle\liminf_{\varepsilon \rightarrow 0} g_\varepsilon(v_\varepsilon)\geq g(v)$. 
Suppose now that $\displaystyle\liminf_{\varepsilon \rightarrow0} g_\varepsilon(v_\varepsilon)<\infty$. 
Passing to a subsequence, if necessary, we can suppose $\displaystyle\liminf_{\varepsilon \rightarrow 0}g_\varepsilon(v_\varepsilon)=
\lim_{\varepsilon \rightarrow 0}g_\varepsilon(v_\varepsilon)<\infty$.

Lemma~\ref{lema01teorema} ensures that $\partial_s v_\varepsilon \rightharpoonup \partial_s v$ and
$\nabla_y v_\varepsilon \rightharpoonup  \nabla_y v$ weakly in $\LL^2(S \times Q)$. As $\tilde{A}_1^\varepsilon(s,y) \rightarrow \la {\bf A}(r(s)), T(s) \ra$ uniformly,  we have 
$$\tilde{\partial}_s v_\varepsilon \rightharpoonup(-i\partial_s- \la {\bf A}(r(s)), T(s) \ra)v$$ 
weakly in $\LL^2(S \times Q)$.
By recalling that $(\tau+\alpha')\in \LL^\infty(S)$, we have
$$ \tilde{\partial}_s v_\varepsilon -  \la \nabla_y v_\varepsilon,  Ry \ra (\tau+\alpha')(s) 
\rightharpoonup (-i\partial_s - \la {\bf A}(r(s)), T(s)\ra )v - \la \nabla_y v, Ry  \ra (\tau+\alpha')(s)$$ 
in $\LL^2(S \times Q)$. By Lemma~\ref{lema01teorema},
$v\in H^1_0(S \times Q)$ and we can write
$v(s,y)=w(s)u_0(y)$ with $\omega\in \hil^1(S)$. 

The  above remarks, together with properties~{\bf (1)} and~{\bf (2)}, show that
\begin{eqnarray*}
\liminf_{\varepsilon\rightarrow0}g_\varepsilon(v_\varepsilon)
& \geq  &
\liminf_{\varepsilon \rightarrow 0}
\int_{S \times Q}
 \frac{1}{\beta^2_\varepsilon}  
\left| \tilde{\partial}_s v_\varepsilon  - \la \nabla_y v_\varepsilon,  R y \ra (\tau+\alpha')(s) \right|^2 {\mathrm d}s{\mathrm d}y \\
& + &
\liminf_{\varepsilon \too 0} \int_{S \times Q} \frac{k^2(s)}{4 \beta_\varepsilon^2} |v_\var|^2 {\mathrm d}s{\mathrm d}y + 
\liminf_{\varepsilon \too 0} \int_{S \times Q} c |v_\var|^2 {\mathrm d}s{\mathrm d}y \\
& \geq & 
\int_{S \times Q} \left|(-i\partial_s - \la {\bf A}(r(s)), T(s) \ra)v 
- \la \nabla_y v,  Ry \ra (\tau+\alpha')(s) \right|^2  {\mathrm d}s{\mathrm d}y\\
& + &
\int_{S \times Q} \left( c - \frac{k^2(s)}{4}\right)|v|^2  {\mathrm d}s{\mathrm d}y \\
& = &
\int_{S} \left\{
\left|(-i\partial_s- \la {\bf A}(r(s)), T(s) \ra) w \right|^2 + 
\left[C(Q) (\tau+\alpha')^2(s) - \frac{k^2(s)}{4} + c\right]| w|^2 \right\} {\mathrm d}s\\
& = &
g_0(\omega)=g(v).
\end{eqnarray*}

Now we are going to show that   for each $v \in \LL^2(S \times Q)$ there exists a sequence  $(v_\varepsilon)_\varepsilon$ in $\LL^2(S \times Q)$ so that $v_\varepsilon  \too v$ in $\LL^2(S \times Q)$ and $\lim_{\varepsilon \too 0} g_\varepsilon (v_\varepsilon) = g(v)$. First we consider the particular case of $v = w u_0$ with $w \in \hil^1(S )$. In this situation we take the sequence $(v_\varepsilon)_\varepsilon$ with, for each $\varepsilon > 0$, $v_\varepsilon := w u_0$. We have $v_\varepsilon\rightarrow v$ in $\LL^2(S \times Q)$ and we see that
\begin{eqnarray*}
& & \lim_{\varepsilon \too 0}
\int_{S \times Q}  (1/\beta^2_\varepsilon) 
\left|   \tilde{\partial}_s (w u_0) - w \la \nabla_y u_0, Ry \ra \, (\tau+\alpha')(s) \right|^2  {\mathrm d}s{\mathrm d}y \\
& = &
\int_{S} \left\{\left| \left(- i \partial_s - \la {\bf A}(r(s)), T(s) \ra  \right) w \right|^2 + C(Q) (\tau+\alpha')^2(s) |w|^2 \right\} {\mathrm d}s
\end{eqnarray*}
and
\begin{eqnarray*}
& &
\lim_{\varepsilon \too 0} \left[\int_{S \times Q} \frac{-k^2(s)}{4 \beta_\varepsilon^2} |wu_0|^2  {\mathrm d}s{\mathrm d}y
+  \int_{S \times Q} c |wu_0|^2 {\mathrm d}s\, {\mathrm d}y \right]\\
& = &
 - \int_{S } \frac{k^2(s)}{4}|w|^2 {\mathrm d}s+ c \int_{S } |w|^2 {\mathrm d}s.
\end{eqnarray*}

Recalling the definitions of $\tilde{A}_2^\varepsilon$,
$\tilde{A}_3^\varepsilon$ and  condition (\ref{nullpotentialcurve}), we obtain
\begin{eqnarray*}
& &
\lim_{\varepsilon \too 0}
\frac{1}{\varepsilon^2} \int_{S \times Q}\left(
\left|-i w \partial_{y_2} u_0 - \tilde{A}_2^\varepsilon w u_0 \right|^2 +
\left|-i w \partial_{y_3} u_0 - \tilde{A}_3^{\varepsilon} w u_0 \right|^2
-
\lambda_0 |wu_0|^2 \right) {\mathrm d}s\, {\mathrm d}y \\
& = &
\lim_{\varepsilon \too 0}
\frac{1}{\varepsilon^2} \int_{S \times Q}
\left[ (\tilde{A}_2^\varepsilon)^2 + (\tilde{A}_3^\varepsilon)^2\right] |w|^2 |u_0|^2  {\mathrm d}s{\mathrm d}y  \\
& = &
\int_{S }\left[
\left( \la N(s), {\bf A}(r(s))\ra \right)^2 + \left(\la B(s), {\bf A}(r(s))\ra \right)^2\right] |w|^2 \, {\mathrm d}s= 0.
\end{eqnarray*}
Thus,
$\lim_{\var \too 0}g_\var(v_\var)=g(wu_0)$.

Now, consider the case  $v \in \LL^2(S \times Q) \backslash \{w u_0: w \in \hil^1(S)\}$. By definition $g(v) = \infty$. Let $v_\varepsilon$ be a sequence so that $v_\varepsilon \rightarrow v$ in $\LL^2(S \times Q)$.  In this case, $\displaystyle \lim_{\varepsilon\rightarrow0}g_\varepsilon(v_\varepsilon)=\infty$.  In fact, if we suppose that $\displaystyle \lim_{\varepsilon \rightarrow0} g_\varepsilon(v_\varepsilon) < \infty$, by Lemmas~\ref{lema01teorema} and ~\ref{lema02teorema} we should have $v = w u_0$, with $ w \in \hil^1(S)$, which does not occur. Therefore, $\displaystyle \lim_{\varepsilon\rightarrow0} g_\varepsilon(v_\varepsilon) = \infty = g(v)$.

We have shown above that the sequence of quadratic forms $g_\var$ $\Gamma$-converges to~$g$ in the strong sense.
To conclude the weak $\Gamma$-convergence we need only to show the following:
if $v_\var \rightharpoonup v $ in $\LL^2(S\times Q)$, then
$\liminf_{\var \to 0} g_\var (v_\var) \geq g(v)$. 
If $\liminf_{\var \to 0} g_\var(v_\var) = \infty$, there is nothing to prove.
If $\liminf_{\var \to 0} g_\var(v_\var) < \infty$, we can suppose that
\[
\liminf_{\var \to 0} g_\var(v_\var) = \lim_{\var \to 0} g_\var(v_\var) <  \infty.
\]
Lemma \ref{lema01teorema} ensures that $(v_\var)_\var$ is bounded in ${\cal H}_0^1(S \times Q)$. By Rellich-Kondrachov Theorem, the space ${\cal H}_0^1(S \times Q)$ is compactly embedded in $\LL^2(S \times Q)$
and so there is a subsequence of $(v_\var)_\var$, also denoted by $(v_\var)_\var$, so that
$v_\var \to v$. 
Now, the proof follows the same steps of strong $\Gamma$-convergence.

The sequence
$g_\var$ $\Gamma$-converges to~$g$ in the strong and weak sense, and so the condition~${\it a)}$ of Proposition~\ref{GammaNorm} in Section~\ref{secIntro} is satisfied.  Since $S \times Q$ is bounded, the operator~$G$ have compact resolvent in $\LL^2(S \times Q)$ and so  item~${\it b)}$ holds true as well.  Again, by the Rellich-Kondrachov Theorem, but now combined with  Lemma~\ref{lema01teorema}, ensure the validity of  item~${\it c)}$.  By applying Proposition~\ref{GammaNorm},  we conclude the proof of Theorem~\ref{teoremaprincipal}.

\begin{Remark}
As usual in the context of (Dirichlet) reduction of dimension, we subtract the diverging coefficient in~\eqref{resultintroductionorm}, but we would like to point that another renormalization,  a ``magnetic renormalization'' by subtracting
\[
\int_{S \times Q} \frac{\lambda_\var(s)}{\varepsilon^2}|v|^2  {\mathrm d}s{\mathrm d}y
\]
from $\tilde{g}_\varepsilon^{\alpha, c}(v)$ would be more natural; here, for each $s \in S$, $\lambda_\varepsilon(s)$ is the first eigenvalue of 
\[
T_\varepsilon^s u  : =  \left[\left(-i\partial_{y_2}- \tilde{A}_2^\varepsilon(s, y)\right)^2 +
\left(-i\partial_{y_3}- \tilde{A}_3^\varepsilon(s, y)\right)^2 \right] u, 
\]
with Dirichlet boundary condition. It it possible to show that as $\varepsilon \rightarrow 0$, $\lambda_\varepsilon (s) \rightarrow  \lambda_0$ uniformly in~$S$, and, under some additional technical hypotheses, the same effective operator~$G_0$ is reached in the limit $\var \to 0$.
\end{Remark}

\subsubsection*{Acknowledgments} We thank an anonymous referee for valuable suggestions. RB was supported by CAPES (Brazil). CRdeO thanks the partial support by CNPq (Brazil).

\

\end{document}